\documentclass[pdflatex,sn-mathphys-num]{sn-jnl}

\usepackage{graphicx}%
\usepackage{multirow}%
\usepackage{amsmath,amssymb,amsfonts}%
\usepackage{amsthm}%
\usepackage{mathrsfs}%
\usepackage[title]{appendix}%
\usepackage{xcolor}%
\usepackage{textcomp}%
\usepackage{manyfoot}%
\usepackage{booktabs}%
\usepackage{algorithm}%
\usepackage{algorithmicx}%
\usepackage{algpseudocode}%
\usepackage{listings}%

\theoremstyle{thmstyleone}%
\newtheorem{theorem}{Theorem}
\newtheorem{proposition}[theorem]{Proposition}%

\theoremstyle{thmstyletwo}%
\newtheorem{remark}{Remark}%

\theoremstyle{thmstylethree}%
\newtheorem{definition}{Definition}%

\begin{document}

\title[Article Title]{Renewal process's guide to fractional Navier--Stokes equations}

\author*{\fnm{Zhang} \sur{Hong}}\email{zhanghong13@cdut.cn}


\abstract{The Navier--Stokes equations, which remain unsolved, are crucial equations in fluid mechanics. Discovering the solutions to the Navier--Stokes equations is one of the challenging Millennium problems. In 1900, Hilbert proposed a potential approach to tackle this problem by establishing the relationship between microscopic dynamics and the macroscopic continuum equations. The key bridge is the derivation of Boltzmann equation and the theory of probability. In this paper, we shall use the collision renewal process with
 arbitrarily distributed waiting times to derive the Boltzmann equation for the time evolution of the probability of the velocity and the displacement of the particle, based on which we prove that the renewal process with exponential collision waiting time is equivalent to the classical Navier--Stokes equations, and that with power--law waiting time is equivalent to the fractional Navier--Stokes equations. Since the collision renewal process with
 arbitrarily distributed waiting times is a random process and is easy to perform  the stochastic simulations of trajectories to obtain the corresponding solution,  we actually find a stochastic approach to solve the classical and fractional Navier--Stokes equations.}




\maketitle

\section{\label{Introduction}Introduction}

The Navier--Stokes (NS) equations are crucial equations that govern the flow of fluids. Even though the equations were formulated in the 19th Century, they have remained unsolved until now. The Clay Mathematics Institute (CMI) in
an official problem description of the NS equation problem mentions:
“There are many fascinating problems and conjectures about the behavior of solutions of the Euler and Navier–Stokes
 equations. ... Standard methods from PDE appear inadequate to settle the problem. Instead, we probably need some deep, new ideas."\cite{F2025}

In 1900 Hilbert proposed the sixth problem,
he says: “The investigations on the foundations of geometry suggest the problem: To treat in the
same manner, by means of axioms, those physical sciences in which already today
mathematics plays an important part; in the first rank are the theory of probabilities
and mechanics." Hilbert suggests a program that aims to
give a rigorous derivation of the macroscopic laws of fluid motion, starting from  the microscopic Newton’s laws on the atomistic level,
using Boltzmann’s kinetic theory as an intermediate step. \cite{H1901, D2025}

Most recently
Deng et al provided a rigorous derivation of Boltzmann’s kinetic equation from the hard-sphere system
for rarefied gas, which is valid for arbitrarily long time as long as the solution to the Boltzmann equation
exists. This is the crucial step towards resolving Hilbert’s sixth problem.
The general strategy follows the paradigm for the long-time derivation of the wave kinetic equation in wave turbulence theory. This is based on propagating a long-time cumulant effect, which keeps memory of the full collision history of the relevant particles, by a partial time expansion. \cite{D2025, DHM2025}

On the other hand, in recent years the NS equations have been extended to the fractional NS equations and regarded as a special case of the latter. \cite{C2015, YWW2025} Compared with the NS equations, the fractional NS equations are capable of handling the memory effects of complex non-Newtonian fluids and internal turbulence during motion. Moreover, they possess more significant advantages in describing anomalous diffusion in fractal media. \cite{ZP2017} Anomalous diffusion is a generalization of the classical Brownian motion.
 The hallmark of anomalous diffusion is the power--law dependence of the mean squared displacement on time, expressed as $\langle\Delta x^2\rangle \sim t^{\beta}$. When the exponent $\beta>1$, the behavior is classified as superdiffusion; for $0<\beta<1$, it is termed subdiffusion; and the case where $\beta=1$ corresponds to normal diffusion (Brownian motion).  One suitable model for describing anomalous diffusion phenomena is the continuous time random walk (CTRW) of which a special case is the Poisson process.
  In recent decades, CTRW and the corresponding master equation approach have been applied to derive the equivalent fractional equations, including the fractional Fokker-Planck equation \cite{MK2000}, the fractional Feymann--Kac equatiion \cite{CB2011,ZLL2013},   the fractional diffusion-reaction equation \cite{SSS2006, AYL2010, HLW2006,F2010}, the
  fractional advection-diffusion-reaction equation \cite{ZL2018,ZL2019},
   and so on. Since the continuous time random walk  is a probability model that facilitates random Monte Carlo simulations of trajectories, it actually provides a  way to solve the equivalent
 fractional equations and obtain their statistical solutions.

  There have been several attempts to establish the equivalence between the equations associated with the fractional NS equations and the CTRW model. For example, using the CTRW or an analogous approach, Friedrich et al. obtained the Kramer--Fokker--Plank equation with a fractional substantial collision derivative, which represents the nonlocal couplings in time and space \cite{FGE2006};  Goychuka generalized the linear Boltzmann equation for the fractional superdiffusive transport of the L$\acute{e}$vy walk type in external force fields, and gave a purely phenomenological fractional BGK
equation without any
well-established trajectory counter-part \cite{G2017}, from which the fractional-order constitutive equations for phonon heat transport were further predicted \cite{LC2019}; and Stokes et al. generalized the Boltzmann equation for nonequilibrium charged particle transport via localized and trapped states \cite{SPCW2016}, etc.
However, to our knowledge, the equivalence of the fractional NS equations from the CTRW model and the corresponding master equations has not been derived since it is difficult for the CTRW model to remember other complex types of interaction (e.g. collisions) in addition to the particle diffusion in the history.

The renewal process is a random process whose time intervals between two renewal steps are arbitrarily distributed. \cite{MO2014} The advantage of the renewal process lies in its ability to retain all contributions from the different kinds of interactions between particles in the history. The renewal process is first used to describe the chemical continuous time random walk (CTRW) by  Aquino et. al. \cite{AD2017}, and then it is employed to derive generalized rate equations for complex chemical reactions and the diffusion of reactants in closed and open heterogeneous chemical--diffusion systems \cite{ZLFL2025}. Additionally, it is used to derive the classical and fractional Hopfield neural networks with complex interactions among neurons \cite{ZLD2026}.
Herein we shall use  the renewal process to derive the
Boltzmann equation from the microscopic dynamics, based on which we obtain the Euler equations and the classical and fractional Navier--Stokes equations by using Chapman--Enskog Method, which is different from the scaling limit method proposed by Deng et al in \cite{DHM2025}. We will show that the renewal process can remember the full historical collisions and prove that the renewal process with exponential distributed collision waiting time is equivalent to the classical Navier--Stokes equations, and that with power--law distributed collision waiting time is equivalent to the fractional Navier--Stokes equations.  Additionally, since the solution of the stochastic simulation of trajectories of the renewal process can be naturally obtained, we actually find an equivalent stochastic method for solving the classical and fractional NS equations. Finally, in this paper, a well-established trajectory counter-part is also given for the phenomenological fractional BGK equation predicted in \cite{G2017}.


\section{The collision renewal process}\label{renewal-process}

We consider the system with $m$ particles that move with random collisions in a $N$-dimensional space. We assume that each of particles has a well-defined diameter $R$ and that all collisions are perfectly elastic.
 Let the position of the particle $i$ be $r_i$ and the velocity of the particle $i$ be  $v_i$, so the state vector is $(r_i,v_i)$.
We assume that there is a random waiting time $t_{i}$ for each particle $i$ to begin to collide with the other particle and the random waiting times $t_{1},t_{2},...,t_{m}$ are all independent for simplicity.
Let $\tau_n=\min\{t_1,t_2,...,t_{m}\}$. Then $\tau_1,...\tau_n$ are independent and identically distributed (i.i.d.) random variables.
If $\tau_n=t_i$, then the nth collision is that  $i$ collides with another particle $j\neq i$ with the measure $\tilde{P}_{ij}=\omega_{ij}|g\cdot k|\pi d^2$, where $g=v_i-v_j$, $k=\frac{x_j-x_i}{|x_i-x_j|}$ and $\omega_{ij}=|x_i-x_j|^{-\gamma}$ is a distance weighting function. Note that the other particle $j$ needs to satisfy the geometrical condition $|x_i-x_j|\leq R$ and the relative motion condition $g\cdot k<0$.
After this collision, according to the assumption of elastic collision, the velocities of $i$ and $j$ become $v'_i=v_i-(k\cdot g)k$, $v'_j=v_j+(k\cdot g)k$, respectively.

\begin{definition}[The collision renewal process]
Let $N(t)$ be the number of collisions until time $t$. Then
$N(t)=\max\{n\leq t: \tau_1+...+\tau_n\leq t\}$ is a renewal process. We call such a process the collision renewal process. The time interval $\tau_n$ is called the renewal waiting time and each collision is called a renewal event (or a renewal step).
\end{definition}

\begin{remark}
Let $S_n=\tau_{i=1}^{n}\tau_i$, then $S_n$ is the arriving time for just taking the nth collision renewal step.
\end{remark}

\subsection{Distribution of the renewal waiting time}
\label{Distribution-waiting-time}

In this section, we will discuss the distribution of the renewal waiting time $\tau_n$.
First, it is easy to find that
\begin{eqnarray}
 P(\tau_n\leq t)&&=1-P(\tau_n\geq t)=1-P(\min\{t_1,t_2,...,t_{m}\}\geq t)\nonumber\\
 &&=1-\prod_{i=1}^{m}
P(t_i\geq t)=1-\prod_{i=1}^{m}
\Psi_i(t).
\label{survivaltimedistribution}
\end{eqnarray}
Here, $\Psi_i(t)=P(t_i\geq t)$ is the survival probability for the particle $i$ not colliding with other particles in the time interval $[0,t]$.
If $\psi_i(t)$ is the probability density function (PDF) of the waiting time $t_i$ for the particle $i$, then
$\Psi_i(t)=\int_t^{\infty}\psi_i(t')dt'$.
Let $\Phi(t)=\prod_{i=1}^{m}
\Psi_i(t)$. Then from Eq.~\eqref{survivaltimedistribution} one can see that $\Phi(t)=1-P(\tau_n\leq t)=P(\tau_n\geq t)$ is the
survival probability that no new renewal events occur in the system in the time interval $[0,t]$.

Let $\phi_i(t)$ be the PDF of the event that the collision of particle $i$ first occurs at time $t$, while none of the other collisions have taken place until time $t$ (which means that the waiting time for the collision of particle $i$ is the minimum waiting time). Then we obtain the following proposition.

\begin{proposition}
Let $\hat{\Phi}(s)$ and $\hat{\phi}_i(s) $ denote the Laplace transforms of $\Phi(t)$ and $\phi_i(t)$, respectively. Here, $i=1,2,...,m$. Then it holds that
  \begin{eqnarray}
  \hat{\Phi}(s)=\frac{1-\sum_{i=1}^{m}\hat{\phi}_i(s)}{s}.
  \label{relationship}
  \end{eqnarray}
\end{proposition}
\begin{proof}
Since $\phi_i(t)$ is the PDF of the event in which the waiting time for the collision of particle $i$ is the minimum waiting time, it can be obtained by differentiating the distribution function for such an event with respect to $t$, that is,
\begin{eqnarray}
\phi_i(t)&&=[P(\tau_n=\tau_i,\tau_n\leq t))]'\nonumber\\
&&=[P(\tau_j\geq \tau_i,\tau_i\leq t)):j=1,..,m,j\neq i]'\nonumber\\
&&=\bigg[\int_0^t \psi_i(t')dt'\int_{t'}^{\infty}\psi_{1}(t^{''})dt^{''}...\int_{t'}^{\infty}\psi_{m}(t^{''})dt^{''}\bigg]'=\psi_i(t)\prod_{j\neq i}\Psi_j(t).
\label{phi}
\end{eqnarray}
On the other side,  we have
\begin{eqnarray}
 \Phi^{'}(t)=-\sum_{i=1}^{m}\bigg[\psi_i(t)\prod_{j\neq i}\Psi_j(t)\bigg]=-\sum_{i=1}^{m}\phi_i(t).
 \label{relationtime}
\end{eqnarray}
Taking the Laplace transform of Eq.~\eqref{relationtime} yields
$u\hat{\Phi}(u)-\Phi(0)=-\sum_{i=1}^{m}\hat{\phi}_i(u)$, which can be changed to Eq.~\eqref{relationship}
 where $\Phi(0)=1$ is used.
\end{proof}

\subsection{The renewal steps of the collision renewal process}
\label{simulationsteps}
In this section, we shall outline the collision renewal steps associated with the simulation of trajectories in the collision renewal process.
The first renewal step includes
the following substeps:

Substep one. The initial positions and velocities of all particles are set at the initial time. This substep only belongs to the first renewal step of the collision renewal process.

Substep two.  The random
waiting times $t_i (i=1,2,...,m)$  as internal clocks are chosen from a series of values distributed
according to $\psi_i(t)$, respectively.
 If the distribution is an exponential distribution $\alpha e^{-\alpha t}$ for $\alpha>0$,  then we can use $-\frac{ln U}{\alpha}$ to get the sample, where $U$ is a random variate drawn from the uniform distribution in the interval $[0,1]$. If the distribution is power--law, i.e., $\beta\tau_0^{\beta}t^{-(1+\beta)}$ for $0<\beta<1$,  then we use
$\tau_0 (1-U)^{-\frac{1}{\beta}}$ to obtain the sample.

Substep three. Find the minimum clock time, namely, $\min\{t_i:i=1,2,...,m\}$. If the minimum waiting time is $t_i$ and $\sum_{j\neq i}\tilde{P}_{ij}\neq 0$, then we will
set the transition probability matrix $(P_{ij})_{m\times m}$ whose components are $P_{ij}=\frac{\tilde{P}_{ij}}{\sum_{j\neq i}\tilde{P}_{ij}}$ for $j\neq i$ and $P_{ii}=0$,
and then choose a special target particle $j$ to collide according to probability $P_{ij}$.

Subtep four. The velocities of $i$ and $j$  are both renewed as below: $v'_i=v_i-(k\cdot g)k$, $v'_j=v_j+(k\cdot g)k$ according to the assumption of elastic collision. The positions of two particles are invariable.

After completing the four substeps, we have accomplished the first renewal step. Subsequently, we initiate the second renewal step by executing only substeps two--four. After that, the renewal cycles are repeated. Note that in this process we ignore the fluctuation effect of the random selection of two equal minimum waiting times and assume that the particles are dense enough to collide (i.e., $\sum_{j\neq i}\tilde{P}_{ij}\neq 0$) for simplicity.
  Note also that according to above renewal steps the Monte Carlo simulation for the collision renewal process can be easily performed.

\section{Generalized master equation: microscopic description}
Let $\vec{x}(t)=(x_1(t),...,x_m(t))$ be the position vector whose component $x_i$ is the position of particle $i$. Let $\vec{v}(t)=(v_1(t),...,v_m(t))$ be the velocity vector whose component $v_i(t)$ is the velocity of the particle $i$ at $t$. Then $\vec{x}(t)$ and $\vec{v}(t)$ constitute a state tensor $(\vec{x}(t),\vec{v}(t))$.
We now investigate the time evolution of the probability of the state tensor in the collision renewal process.
Let $P(\vec{x},\vec{v}, t)$ denote the probability distribution $P(\vec{x}(t)=\vec{x},\vec{v}(t)=\vec{v})$. Here, $\vec{x}=(x_1,...,x_m)$ and $\vec{v}=(v_1,...,v_m)$ are  position and  velocity vectors whose components are constant.
First, according to renewal
theory, the probability that the system is in state $(\vec{x},\vec{v})$ at time $t$
is equal to the probability that the system just reaches state $(\vec{x},\vec{v})$ at an earlier time $t'<t$ and remains
in that state until $t$.
If $R_{n}(\vec{x},\vec{v},t)$ is the PDF of just arriving at the state $(\vec{x},\vec{v})$ at time $t$ after $n$ renewal steps, then one has
\begin{eqnarray}
    P(\vec{x},\vec{v},t)=\sum_{n=1}^{\infty}\int_0^t R_n(x-\vec{v}(t-t'),\vec{v},t')\Phi(t-t')dt'.
    \label{eq:renewalbalanceequation1}
\end{eqnarray}
Note that this is a balance equation for the collision renewal process.

In addition, since the random event that the system arrives at state $(\vec{x},\vec{v})$ at time $t$ through $(n + 1)$th step renewal is equivalent to the random event that the state vector arrives at $(x-(\vec{v}-\vec{v_{ij}})(t-t'),\vec{v}-\vec{v_{ij}})$ at an earlier time $t'<t$ through $n$th step renewal (that is, $S_n=t'<t$), and after a time interval of $t - t'$ to take the next renewal (i.e., the collision of particle $i$ with $j$) at $t$ (that is, $S_{n+1}=t$) causing the state to change to $(\vec{x},\vec{v})$, we can obtain the other balance equation for the collision renewal process as following,

\begin{eqnarray}
    R_{n+1}(\vec{x},\vec{v},t)=\sum_{j\neq i}\sum_{i=1}^{m}\int_0^t R_n(x-(\vec{v}-\vec{v_{ij}})(t-t'),\vec{v}-\vec{v_{ij}},t')\phi_i(t-t')P_{ij}dt'.
    \label{eq:renewalbalanceequation2}
\end{eqnarray}
 Here,  $\vec{v_{ij}}=(0,...,(g\cdot k)k,...,-(g\cdot k)k,...0)$ (i.e., the $i$th component is $(g\cdot k)k$, $j$th component is $-(g\cdot k)k$, and the left components are all $0$).

 \begin{theorem}[Generalized master equation]\label{materequationthm}
 Let $R(\vec{x},\vec{v},t)=\sum_{n=0}^{\infty}R_n(\vec{x},\vec{v},t)$ where $R_0(\vec{x},\vec{v},t)=P(\vec{x},\vec{v},0)\delta(t)$ and let
\begin{eqnarray}
    \hat{\Theta}_i(s)=\frac{\hat{\phi}_i(s)}{\hat{\Phi}(s)},
    \label{Theta}
\end{eqnarray}
 where $\hat{f}(s)$ denotes the Laplace transform of $f(t)$.
Then one can obtain the generalized master equation for the time evolution for the probability $P(\vec{x},\vec{v},t)$ as following
\begin{eqnarray}
 \vec{v}\cdot\nabla P(\vec{x},\vec{v},t)+\frac{\partial P(\vec{x},\vec{v},t)}{\partial t}
 &&=\sum_{j\neq i}\sum_{i=1}^{m}\int_0^t P(\vec{x}-(\vec{v}-\vec{v_{ij}})(t-t'),\vec{v}-\vec{v_{ij}},t')\Theta_i(t-t')P_{ij}dt'\nonumber\\
&&-\sum_{i=1}^{m}P(\vec{x}-\vec{v}(t-t'),\vec{v},t')\Theta_i(t-t')dt'.
 \label{master}
 \end{eqnarray}
\end{theorem}

\begin{proof}
 From the first balance equation \eqref{eq:renewalbalanceequation1} we obtain
\begin{eqnarray}
P(\vec{x},\vec{v},t)=\int_0^t R(\vec{x}-\vec{v}(t-t'),\vec{v},t)\Phi(t-t')dt'.
\label{P-R}
\end{eqnarray}
 From the second balance equation \eqref{eq:renewalbalanceequation2} one has
\begin{eqnarray}
&&R(\vec{x},\vec{v},t)-R_0(\vec{x},\vec{v},t)=\sum_{j\neq i}\sum_{i=1}^{m}\int_0^t R(\vec{x}-(\vec{v}-\vec{v_{ij}})(t-t'),\vec{v}-\vec{v_{ij}},t)\phi_i(t-t')P_{ij}dt'.~~~~~~~
\label{R-R}
\end{eqnarray}
Taking Fourier $\vec{x}\rightarrow \vec{k}$ and Laplace $t\rightarrow s$ transforms of Eq.~\eqref{P-R} yields
\begin{eqnarray}
\hat{\hat{P}}(\vec{k},\vec{v},s)=\hat{\hat{R}}(\vec{k},\vec{v},s)\Phi(u+i\vec{k}\cdot \vec{v}).
\label{P-R-laplace-fourier}
\end{eqnarray}
where the function  $\hat{\hat{f}}(\vec{k},s)$ denotes the Fourier-Laplace transform of $f(\vec{x},t)$.
By taking the Fourier $\vec{x}\rightarrow \vec{k}$ transform of Eq.~\eqref{R-R}, we find
\begin{eqnarray}
&&\hat{R}(\vec{k},\vec{v},t)-\hat{R}_0(\vec{k},\vec{v},t)=\sum_{j\neq i}\sum_{i=1}^{m}\int_0^t \hat{R}(\vec{k},\vec{v}-\vec{v_{ij}},t)\phi_i(t-t')e^{-i\vec{k}\cdot(\vec{v}-\vec{v_{ij}})(t-t')}P_{ij}dt'.~~
\label{R-R-Fourier}
\end{eqnarray}
Here, $\hat{f}(\vec{k})$ is the Fourier transform of $f(\vec{x})$.
Furthermore, we take the Laplace transform of Eq.~\eqref{R-R-Fourier} and combine with $R_0(\vec{x},\vec{v},t)=P(\vec{x},\vec{v},0)\delta(t)$, and obtain
\begin{eqnarray}
&&\hat{\hat{R}}(\vec{k},\vec{v},s)-\hat{P}(\vec{k},\vec{v},0)=\sum_{j\neq i}\sum_{i=1}^{m}\hat{\hat{R}}(\vec{k},\vec{v}-\vec{v_{ij}},s)\hat{\phi}_i[s+i\vec{k}\cdot(\vec{v}-\vec{v_{ij}}]P_{ij}.
\label{R-R-Fourier-laplace}
\end{eqnarray}
By using Eqs.~\eqref{relationship}, \eqref{P-R-laplace-fourier} and \eqref{R-R-Fourier-laplace}, we can rewrite $i\vec{k}\cdot\vec{v}\hat{\hat{P}}(\vec{k},\vec{v},s)+[s\hat{\hat{P}}(\vec{k},\vec{v},s)-\hat{P}(\vec{k},\vec{v},0)]$ as following
\begin{eqnarray}
i\vec{k}\cdot\vec{v}\hat{\hat{P}}(\vec{k},\vec{v},s)+[s\hat{\hat{P}}(\vec{k},\vec{v},s)-\hat{P}(\vec{k},\vec{v},0)]
&&=(i\vec{k}\cdot\vec{v}+s)\hat{\hat{P}}(\vec{k},\vec{v},s)-\hat{P}(\vec{k},\vec{v},0)\nonumber\\
&&=(i\vec{k}\cdot\vec{v}+s)\hat{\hat{R}}(\vec{k},\vec{v},s)\Phi(s+i\vec{k}\cdot \vec{v})-\hat{P}(\vec{k},\vec{v},0)\nonumber\\
&&=(i\vec{k}\cdot\vec{v}+s)\hat{\hat{R}}(\vec{k},\vec{v},s)\frac{1-\sum_{i=1}^{m}\hat{\phi}_i(s+i\vec{k}\cdot \vec{v})}{s+i\vec{k}\cdot \vec{v}}-\hat{P}(\vec{k},\vec{v},0)\nonumber\\
&&=\hat{\hat{R}}(\vec{k},\vec{v},s)-\sum_{i=1}^{m}\hat{\hat{R}}(\vec{k},\vec{v},s)\hat{\phi}_i(s+i\vec{k}\cdot \vec{v})-\hat{P}(\vec{k},\vec{v},0)\nonumber\\
&&=\sum_{j\neq i}\sum_{i=1}^{m}\hat{\hat{R}}(\vec{k},\vec{v}-\vec{v_{ij}},s)\hat{\phi}_i[s+i\vec{k}\cdot(\vec{v}-\vec{v_{ij}})]P_{ij}+\hat{P}(\vec{k},\vec{v},0)\nonumber\\
&&-\sum_{i=1}^{m}\hat{\hat{R}}(\vec{k},\vec{v},s)\hat{\phi}_i(s+i\vec{k}\cdot \vec{v})-\hat{P}(\vec{k},\vec{v},0)\nonumber\\
&&=\sum_{j\neq i}\sum_{i=1}^{m}\hat{\hat{P}}(\vec{k},\vec{v}-\vec{v_{ij}},s)\frac{\hat{\phi}_i[s+i\vec{k}\cdot(\vec{v}-\vec{v_{ij}})]}{\hat{\Phi}[s+i\vec{k}\cdot(\vec{v}-\vec{v_{ij}})]}P_{ij}\nonumber\\
&&-\sum_{i=1}^{m}\hat{\hat{P}}(\vec{k},\vec{v},s)\frac{\hat{\phi}_i(s+i\vec{k}\cdot \vec{v})}{\hat{\Phi}(s+i\vec{k}\cdot \vec{v})}
\label{masterlaplacefourier}
\end{eqnarray}
Taking the inverse Laplace transform of Eq.~\eqref{masterlaplacefourier} yields
\begin{eqnarray}
 i\vec{k}\cdot\vec{v}\hat{P}(\vec{k},\vec{v},t)+\frac{\partial \hat{P}(\vec{k},\vec{v},t)}{\partial t} &&=\sum_{j\neq i}\sum_{i=1}^{m}\int_0^t\hat{P}(\vec{k},\vec{v}-\vec{v_{ij}},t')\Theta_i(t-t')e^{-i\vec{k}\cdot(\vec{v}-\vec{v_{ij}})(t-t')}P_{ij}dt'\nonumber\\
&&-\sum_{i=1}^{m}\int_0^t\hat{P}(\vec{k},\vec{v},t')\Theta_i(t-t')e^{-i\vec{k}\cdot\vec{v}(t-t')}dt'.
 \label{masterFourier}
\end{eqnarray}
We then take the inverse Fourier transform of Eq.~\eqref{masterFourier} and obtain the generalized master equation \eqref{master} for the collision renewal process.
 \end{proof}

\section{Generalized rate equation: mesoscopic description}
We now focus on one particle and obtain the corresponding generalized rate equations based on the derived master equation \eqref{master}.

Let $n_l(x_a, v_b, t)=\sum_{\vec{x}:x_l=x_a, \vec{v}: v_l=v_b} P(\vec{x},\vec{v},t)$ be the probability for the particle $l$ whose position is $x_a$ and whose velocity is $v_b$. We assume that the states of all particles at time $t$ are independent. Then we find
\begin{eqnarray}
 n_l(x_a, v_b, t)n_j(x'_a, v'_b, t)
&&  = \sum_{\vec{x}:x_l=x_a, \vec{v}: v_l=v_b} P(\vec{x},\vec{v},t)\sum_{\vec{x}:x_j=x'_a, \vec{v}: v_j=v'_b} P(\vec{x},\vec{v},t)\nonumber\\&&=\sum_{\vec{x}:x_l=x_a, x_j=x'_a; \vec{v}: v_l=v_b, v_j=v'_b} P(\vec{x},\vec{v},t),
  \label{independence}
\end{eqnarray}
for any $j\neq l$.
Additionally, one has
\begin{eqnarray}
 v_b\frac{\partial n_l(x_a,v_b,t)}{\partial x_a}   &&= \sum_{\vec{x}:x_l=x_a, \vec{v}: v_l=v_b} v_b \frac{\partial P(\vec{x},\vec{v},t)}{\partial x_a} \nonumber\\
  &&=\sum_{\vec{x}:x_l=x_a, \vec{v}: v_l=v_b} v_l \frac{\partial P(\vec{x},\vec{v},t)}{\partial x_a}+\sum_{j\neq i}\sum_{\vec{x}:x_l=x_a, \vec{v}: v_l=v_b} v_j \frac{\partial P(\vec{x},\vec{v},t)}{\partial x_j}\nonumber\\
  &&=\sum_{\vec{x}:x_l=x_a, \vec{v}: v_l=v_b}\vec{v}\cdot\nabla P(\vec{x},\vec{v},t).
  \label{oneterm}
\end{eqnarray}
In the second equation we used the fact $\sum_{\vec{x}:x_l=x_a, \vec{v}: v_l=v_b} v_j \frac{\partial P(\vec{x},\vec{v},t)}{\partial x_j}= v_j\frac{\partial n_l(x_a,v_b,t)}{\partial x_j}=0$.
Therefore, from Eqs.~\eqref{master}, \eqref{independence} and \eqref{oneterm}, we find
\begin{eqnarray}
 v_b\frac{\partial n_l(x_a,v_b,t)}{\partial x_a}+\frac{\partial n_l(x_a,v_b,t)}{\partial t}
 &&=\sum_{\vec{x}:x_l=x_a, \vec{v}: v_l=v_b}\bigg[v_b\frac{\partial P(\vec{x},\vec{v},t)}{\partial x_a}+\frac{\partial P(\vec{x},\vec{v},t)}{\partial t}\bigg]\nonumber\\
 &&=\sum_{\vec{x}:x_l=x_a, \vec{v}: v_l=v_b}\bigg[\vec{v}\cdot\nabla P(\vec{x},\vec{v},t)+\frac{\partial P(\vec{x},\vec{v},t)}{\partial t}\bigg]\nonumber\\
 &&=\sum_{j\neq i}\int\int dx'_a dv'_b\bigg\{\sum_{\vec{x}:x_l=x_a, \vec{v}: v_l=v_b;x_j=x'_a, \vec{v}: v_j=v'_b}\bigg[\sum_{i=1}^{m}\nonumber\\
 &&\int_0^t P(\vec{x}-(\vec{v}-\vec{v_{ij}})(t-t'),\vec{v}-\vec{v_{ij}},t')\Theta_i(t-t')P_{ij}dt'\nonumber\\
&&-\sum_{i=1}^{m}P(\vec{x}-\vec{v}(t-t'),\vec{v},t')\Theta_i(t-t')dt'\bigg]\bigg\}\nonumber\\
&&=\sum_{j\neq i}\int\int dx'_a dv'_b\bigg[\int_0^t n_l\bigg(x_a-[v_b\nonumber\\
&&-(k\cdot (v_b-v'_b))k](t-t'),v_b-(k\cdot (v_b-v'_b))k,t'\bigg)\nonumber\\
&&\cdot n_j\bigg(x'_a-[v'_b+(k\cdot (v_b-v'_b))k](t-t'),v'_b\nonumber\\
&&+(k\cdot (v_b-v'_b))k,t'\bigg)\bigg(\Theta_l(t-t')P_{lj}+\Theta_j(t-t')P_{jl}\bigg)\nonumber\\
&& -\int_0^t n_l\bigg(x_a-v_b(t-t'),v_b,t'\bigg)\nonumber\\
&&\cdot n_j\big(x'_a-v'_b(t-t'),v'_b,t'\big)\bigg(\Theta_l(t-t')P_{lj}\nonumber\\
&&+\Theta_j(t-t')P_{jl}\bigg)dt'\bigg].
\end{eqnarray}
Here, in the last equation we consider both the effects of the collisions of $l\rightarrow j$ ($l$ collides with $j$) and $j\rightarrow l$ ($j$ collides with $i$). To sum up, we obtain the following theorem for the generalized rate equation.
\begin{theorem}[Generalized rate equation]\label{rateeqationthm}
If $n_l(x_a, v_b, t)$ is the probability for the particle $l$ whose position is $x_a$ and whose velocity is $v_b$, then the generalized rate equation
\begin{eqnarray}
 v_b\frac{\partial n_l(x_a,v_b,t)}{\partial x_a}+\frac{\partial n_l(x_a,v_b,t)}{\partial t}
 &&=\sum_{j\neq i}\int\int dx'_a dv'_b\bigg[\int_0^t n_l\bigg(x_a-[v_b\nonumber\\
&&-(k\cdot (v_b-v'_b))k](t-t'),v_b-(k\cdot (v_b-v'_b))k,t'\bigg)\nonumber\\
&&\cdot n_j\bigg(x'_a-[v'_b+(k\cdot (v_b-v'_b))k](t-t'),v'_b\nonumber\\
&&+(k\cdot (v_b-v'_b))k,t'\bigg)\bigg(\Theta_l(t-t')P_{lj}+\Theta_j(t-t')P_{jl}\bigg)\nonumber\\
&& -\int_0^t n_l\bigg(x_a-v_b(t-t'),v_b,t'\bigg)\nonumber\\
&&\cdot n_j\bigg(x'_a-v'_b(t-t'),v'_b,t'\bigg)\bigg(\Theta_l(t-t')P_{lj}\nonumber\\
&&+\Theta_j(t-t')P_{jl}\bigg)dt'\bigg].
\label{rateequation}
\end{eqnarray}
holds.
\end{theorem}
Note that if we just consider the effect of the collision of $l\rightarrow j$, then the generalized rat equation becomes the simpler form
\begin{eqnarray}
 v_b\frac{\partial n_l(x_a,v_b,t)}{\partial x_a}+\frac{\partial n_l(x_a,v_b,t)}{\partial t}
 &&=\sum_{j\neq i}\int\int dx'_a dv'_b\bigg[\int_0^t n_l\bigg(x_a-[v_b\nonumber\\
 &&-(k\cdot (v_b-v'_b))k](t-t'),v_b-(k\cdot (v_b-v'_b))k,t'\bigg)\nonumber\\
&&\cdot n_j\bigg(x'_a-[v'_b+(k\cdot (v_b-v'_b))k](t-t'),v'_b\nonumber\\
&&+(k\cdot (v_b-v'_b))k,t'\bigg)\Theta_l(t-t')P_{lj}\nonumber\\
&& -\int_0^t n_l\bigg(x_a-v_b(t-t'),v_b,t'\bigg)\nonumber\\
&&\cdot n_j\bigg(x'_a-v'_b(t-t'),v'_b,t'\bigg)\Theta_l(t-t')P_{lj}dt'\bigg].
\label{rateeqautionsimple}
\end{eqnarray}
In the next section we shall just consider the effect of the collision of $l\rightarrow j$ for simplicity.

\section{Mesoscopic classical and fractional BGK equations}
In this section we shall consider two special cases of the collision renewal process and derived the corresponding classical and fractional BGK equations from the simple rate equation \eqref{rateeqautionmoresimple}.
\subsection{Exponential case and classical Boltzmann equation and BGK equation}
We first consider the special case of the collision renewal process with the exponential waiting time PDF.
\begin{theorem}[Classical Boltzmann equation]\label{Boltzmannthm1}
If in the collision renewal process the PDF for the waiting time $t_i$ is $\psi_i(t)=r_ie^{-r_i t} (i=1,2,...,m)$, then the classical Boltzmann equation
\begin{eqnarray}
 v_b\frac{\partial n_l(x_a,v_b,t)}{\partial x_a}+\frac{\partial n_l(x_a,v_b,t)}{\partial t} &&=\sum_{j\neq i}\int\int \bigg[ n_l(x_a,v_b-(k\cdot (v_b-v'_b))k,t)\nonumber\\
&&\cdot n_j(x'_a,v'_b+(k\cdot (v_b-v'_b))k,t) r_l P_{lj}\nonumber\\
&& - n_l(x_a,v_b,t)n_j(x'_a,v'_b,t)r_l P_{lj}\bigg]dx'_a dv'_b.
\label{Boltzmannequation}
\end{eqnarray}
holds.
\end{theorem}
\begin{proof}
 If $\psi_i(t)=r_ie^{-r_i t}$, then one has
\begin{eqnarray}
\Psi_{i}(t)=\int_t^{\infty}\psi_i(t')dt'=e^{-r_{i}t}.
\end{eqnarray}
Thus,
\begin{eqnarray}
\Phi(t)=\prod_{i=1}^{m}
\Psi_i(t)=\exp\bigg(-\sum_{i=1}^{m}r_{i}t\bigg).
\end{eqnarray}
 In addition, from Eq.~\eqref{phi}, one finds
\begin{eqnarray}
\phi_{i}(t)&=&\psi_i(t)\prod_{j\neq i}
\Psi_j(t)=r_i\exp\bigg(-\sum_{i=1}^{m}r_{i}t\bigg),
\end{eqnarray}
and then
\begin{eqnarray}
  \hat{\Theta}_{l}(s)=\frac{\phi_{l}(s)}{\Phi(s)}=r_l.
\end{eqnarray}
Inverting it to time space yields
\begin{eqnarray}
 \Theta_{l}(t)=r_l\delta(t).
 \label{exponentialtheta}
\end{eqnarray}
Substituting Eq.~\eqref{exponentialtheta} into Eq.~\eqref{rateeqautionsimple} yields Eq.\eqref{Boltzmannequation}.
\end{proof}
Furthermore, according to the BGK approximations $n^{eq}_l (x_a, v_b, t)\approx n_l(x_a,v_b-(k\cdot (v_b-v'_b))k,t)$ and $\sum_{j\neq i}P_{lj}=1$ where $n^{eq}_l$ is the local equilibrium distribution, and $\int\int n_j(x'_a,v'_b+(k\cdot (v_b-v'_b))k,t) dx'_a dv'_b=\int\int n_j(x'_a,v'_b,t)dx'_a dv'_b=1$,  Eq.~\eqref{Boltzmannequation} becomes
\begin{eqnarray}
&& v_b\frac{\partial n_l(x_a,v_b,t)}{\partial x_a}+\frac{\partial n_l(x_a,v_b,t)}{\partial t}\approx r_l  [n^{eq}_l(x_a,v_b,t)- n_l(x_a,v_b,t)].~~~~~~
\label{rateeqautionsimpleexponent}
\end{eqnarray}
One can see that Eq.~\eqref{rateeqautionsimpleexponent} is the classical BGK equation
\begin{eqnarray}
 && v_b\frac{\partial n_l(x_a,v_b,t)}{\partial x_a}+\frac{\partial n_l(x_a,v_b,t)}{\partial t}\approx -\frac{1}{\tau} [ n_l(x_a,v_b,t)-n^{eq}_l(x_a,v_b,t)].
\label{BGK}
\end{eqnarray}
Here, $\tau=\frac{1}{r_l}$ is the  mean collision time of particle $l$.

\subsection{Power--law case and fractional BGK equation}
If in the collision renewal process the PDF for the waiting time $t_l$ is  $\psi_{l}(t)\sim\tau_0^{\beta}\beta\frac{1}{t^{1+\beta}}$ where $0<\beta<1$, and the PDF for $t_j$ is $\psi_{j}(t)=r_j e^{-r_j t}$  for $ j\neq l$.
Then
 $\Psi_{l}(t)\sim\tau_0^{\beta}\beta\frac{1}{t^{\beta}}$ and $\Psi_{j}(t)=e^{-r_j t}$,
and thus
$ \Phi(t)\sim \tau_0^{\beta}\beta\frac{1}{t^{\beta}}\exp(-\sum_{j\neq l}r_j t)),$
and
 $\phi_l(t)\sim\tau_0^{\beta}\beta\frac{1}{t^{1+\beta}}\cdot \exp(-\sum_{j\neq l}^{m}r_j t)$.
In the Laplace space we find \cite{ZLFL2025}
\begin{eqnarray}
\hat{\Theta}_{l}(s)\sim\frac{1}{\Gamma(1-\beta)\tau_0^{\beta}}(s+\alpha)^{1-\beta},
\label{thetapowerlaw}
\end{eqnarray}
where $\alpha=\sum_{j\neq l}r_j$.

We assume that the effect of the collision is very small and use  the BGK approximations $n^{eq}_l (x_a, v_b, t')\approx n_l(x_a-[v_b-(k\cdot (v_b-v'_b))k](t-t'),v_b-(k\cdot (v_b-v'_b))k,t')$ and $n_l(x_a,v_b,t')\approx n_l(x_a-v_b(t-t'),v_b,t')$, and $\sum_{j\neq i}P_{lj}= 1$, and the normalization conditions $\int\int
 n_j(x'_a-[v'_b+(k\cdot (v_b-v'_b))k](t-t'),v'_b+(k\cdot (v_b-v'_b))k,t')dx'_a dv'_b=\int\int n_j(x'_a-v'_b(t-t'),v'_b,t')dx'_a dv'_b=1$, then Eq.~\eqref{rateeqautionsimple}  becomes
\begin{eqnarray}
 v_b\frac{\partial n_l(x_a,v_b,t)}{\partial x_a}+\frac{\partial n_l(x_a,v_b,t)}{\partial t}
 &&\approx\int_0^t n^{eq}_l(x_a,v_b,t')\Theta_l(t-t')dt'\nonumber\\
&& -\int_0^t n_l(x_a,v_b,t')\Theta_l(t-t')dt'.
\label{rateeqautionmoresimple}
\end{eqnarray}
Taking Laplace transform of Eq.~\eqref{rateeqautionmoresimple} yields
\begin{eqnarray}
 && v_b\frac{\partial \hat{n}_l(x_a,v_b,s)}{\partial x_a}+s\hat{n}_l(x_a,v_b,u)-n_l(x_a, v_b,0)\nonumber\\
 &&\approx \big(\hat{n}^{eq}_l(x_a,v_b,s)-\hat{n}_l(x_a,v_b,s)\big)\hat{\Theta}_l(s).
\label{rateeqautionmoresimple2}
\end{eqnarray}
We substitute Eq.~\eqref{thetapowerlaw} into Eq.~\eqref{rateeqautionmoresimple2} and take the inverse Laplace transform and obtain the fractional BGK equation as following
\begin{eqnarray}
 v_b\frac{\partial n_l(x_a,v_b,t)}{\partial x_a}+\frac{\partial n_l(x_a,v_b,t)}{\partial t}&&\approx \frac{1}{\Gamma(1-\beta)\tau_0^{\beta}}e^{-\alpha t}D_t^{1-\beta}\bigg\{e^{\alpha t}[n^{eq}_l(x_a,v_b,t)\nonumber\\
 &&- n_l(x_a,v_b,t')]\bigg\},
\label{fractionalBGKequation}
\end{eqnarray}
where $e^{-\alpha t}D_{t}^{1-\beta} \big(e^{\alpha t}f(t)\big)$ is a fractional derivative operator \cite{ZL2018}, defined by
\begin{eqnarray}
e^{-\alpha t}D_{t}^{1-\beta} \big(e^{\alpha t}f(t)\big)&&=\frac{1}{\Gamma(1-\beta)}\bigg(\frac{d}{dt}\int_{0}^{t}e^{-\alpha(t-t')}\frac{f(t')}{(t-t')^{\beta}}dt'\nonumber\\
&&+\alpha \int_{0}^{t}e^{-\alpha(t-t')}\frac{f(t')}{(t-t')^{\beta}}dt' \bigg),
\label{RLgenerated}
\end{eqnarray}
whose Laplace transform satisfies
\begin{eqnarray}
L\{e^{-\alpha t}D^{1-\beta}(e^{\alpha t} f(t))\}=(s+\alpha)^{1-\beta}f(s).~~
\label{RLgeneratedlaplace}
\end{eqnarray}
Here, $L(f(t))$ denotes the Laplace transform of $f(t)$. Note that the fractional operator \eqref{RLgenerated} can reduce to the Riemann--Liouville fractional derivative operator  when $\alpha=0$ \cite{P1999}.  From fractional BGK equation \eqref{fractionalBGKequation} one can see that the time evolution of the probability of partilce $l$ with position $x_a$ and velocity $v_b$ at time $t$ has fractional memory of the history and depends on the collisions with other particles. This fractional memory comes from the power--law distribution of collision waiting times of the particle $l$.
When $\alpha=0$, Eq.~\eqref{fractionalBGKequation} reduces to the phenomenological fractional BGK equation
for $F=0$ predicted in \cite{G2017}.

\section{Chapman--Enskog method and the macroscopic fractional Navier--Stokes equations: }
As we all know, by using the Chapman--Enskog method, the BGK equation \eqref{BGK} can be reduced to the classical Euler equations and Navier--Stokes equations \cite{CC1970}. We herein shall use the Chapman--Enskog method and the fractional BGK equation \eqref{fractionalBGKequation} to obtain fractional Navier--Stokes equations. We will consider the simplest 1D case of the collision renewal process, and the results can be easily extended to multidimensional space.

\begin{theorem}[Fractional Navier--Stokes equations]\label{fractionalnsequationthm}
If in the collision renewal process the PDF for the waiting time $t_l$ is $\psi_{l}(t)\sim\tau_0^{\beta}\beta\frac{1}{t^{1+\beta}}$ where $0<\beta<1$, and the PDF for $t_j$ is $\psi_{j}(t)=r_j e^{-r_j t}$  for $ j\neq l$. Then one can otain the fractional Navier--Stokes equations
\begin{eqnarray}
    \frac{\partial \rho} {\partial t}+\frac{\partial (\rho u)} {\partial x}=0,
 \label{NS1}
\end{eqnarray}
\begin{eqnarray}
  \frac{\tau}{\Gamma(1-\beta)\tau_0^{\beta} }e^{-\alpha t}D_t^{1-\beta}\bigg[e^{\alpha t} \bigg(\frac{\partial (\rho u)} {\partial t}+\frac{\partial E} {\partial x}\bigg)\bigg]+\frac{\partial \tilde{p}^{(1)}} {\partial x}=0,~~~~~
 \label{NS2}
\end{eqnarray}
\begin{eqnarray}
    \frac{\tau}{\Gamma(1-\beta)\tau_0^{\beta} }e^{-\alpha t}D_t^{1-\beta}\bigg[e^{\alpha t} \bigg(\frac{\partial E}{\partial t}+\frac{\partial u(E+p)} {\partial x}\bigg)\bigg] +\frac{\partial \tilde{q}}{\partial x}=0.
    \label{NS3}
\end{eqnarray}
Here, $n=\int f^{(0)}dv$, $\rho=m\int f^{(0)}dv$, $u=\frac{1}{n}\int f^{(0)}vdv$, $T=\frac{m}{n K_B}\int (v-u)^2 f^{(0)}dv$, $E=\frac{1}{2}\int m v^2 f^{(0)}dv=\frac{1}{2}nK_BT+\frac{1}{2}\rho u^2$ and $P=\int m(v-u)^2 f^{(0)}dv=n K_B T$ where $m$ is the mass of the particle and $K_B$ is the Boltzmann constant. In addition,
$\tilde{p}^{(1)}=\int mv^2 \tilde{f}^{(1)}dv$ is the classical viscous stress and  $\tilde{q}=\frac{1}{2}\int mv^3 \tilde{f}^{(1)}dv$ is the classical heat loss quantity  with  $\tilde{f}^{(1)}=-\tau[\frac{\partial f^{(0)}} {\partial t_1}+v\frac{\partial f^{(0)}} {\partial x_1}]$.
\end{theorem}
\begin{proof}
In this proof we shall respectively use $f(x,v,t)$ and $f^{eq}(x,v,t)$ to denote the densities $n_0\cdot n_l(x_a, v_b, t)$ and $n_0\cdot n^{eq}_l(x_a, v_b, t)$ for simplicity. Here, $n_0$ is the number of particles.
Let
\begin{eqnarray}
 f\sim f^{(0)}+\epsilon f^{(1)}+\epsilon^2 f^{(2)},
 \label{Chapman-Enskogassumpttion1}
\end{eqnarray}
\begin{eqnarray}
 \frac{\partial f}{\partial t}=\frac{\partial f}{\partial t_0}+\epsilon \frac{\partial f}{\partial t_1}+\epsilon^2 \frac{\partial f}{\partial t_2},
 \label{Chapman-Enskogassumpttion2}
\end{eqnarray}
and
\begin{eqnarray}
  \frac{\partial f}{\partial x}=\epsilon \frac{\partial f}{\partial x_1},
  \label{Chapman-Enskogassumpttion3}
\end{eqnarray}
for small $\epsilon$ where $f^{(0)}=f^{eq}$.
We substitute Eqs.~\eqref{Chapman-Enskogassumpttion1}, \eqref{Chapman-Enskogassumpttion2} and \eqref{Chapman-Enskogassumpttion3} into Eq.~\eqref{fractionalBGKequation}, and obtain
\begin{eqnarray}
  &&  \bigg(\frac{\partial f^{(0)}}{\partial t_0}+\epsilon \frac{\partial f^{(0)}}{\partial t_1}+\epsilon^2 \frac{\partial f^{(0)}}{\partial t_2}\bigg)+\epsilon \bigg(\frac{\partial f^{(1)}}{\partial t_0}+\epsilon \frac{\partial f^{(1)}}{\partial t_1}+\epsilon^2 \frac{\partial f^{(1)}}{\partial t_2}\bigg) \nonumber\\
  &&+\epsilon^2\bigg(\frac{\partial f^{(2)}}{\partial t_0}+\epsilon \frac{\partial f^{(2)}}{\partial t_1}+\epsilon^2 \frac{\partial f^{(2)}}{\partial t_2}\bigg)+v\bigg(\epsilon\frac{\partial f^{(0)}}{\partial x_1}+\epsilon^2 \frac{\partial f^{(1)}}{\partial x_1}+\epsilon^3 \frac{\partial f^{(2)}}{\partial x_1}\bigg)\nonumber\\
  &&=-\frac{1}{\Gamma(1-\beta)\tau_0^{\beta}}e^{-\alpha t}D_t^{1-\beta}\bigg\{e^{\alpha t}\bigg[\epsilon f^{(1)}+\epsilon^2 f^{(2)}\bigg]\bigg\}.
  \end{eqnarray}
Since the terms with the same power of $\epsilon$ are equal, we obtain
\begin{eqnarray}
    \frac{\partial f^{(0)}} {\partial t_0}=0,
 \label{zeroorder}
\end{eqnarray}
for $\epsilon^0$,
\begin{eqnarray}
    &&\frac{\partial f^{(0)}} {\partial t_1}+v\frac{\partial f^{(0)}} {\partial x_1}=-\frac{1}{\Gamma(1-\beta)\tau_0^{\beta}}e^{-\alpha t}D_t^{1-\beta}(e^{\alpha t} f^{(1)}),
 \label{1order}
\end{eqnarray}
for $\epsilon^1$,
\begin{eqnarray}
    &&\frac{\partial f^{(0)}} {\partial t_2}+\frac{\partial f^{(1)}} {\partial t_1}+\frac{\partial f^{(2)}} {\partial t_0}+\frac{\partial f^{(1)}} {\partial x_1}=-\frac{1}{\Gamma(1-\beta)\tau_0^{\beta}}e^{-\alpha t}D_t^{1-\beta}(e^{\alpha t}f^{(2)}),
 \label{2order}
\end{eqnarray}
for $\epsilon^2$. In Eq.~\eqref{1order} we assume $\frac{\partial f^{(1)}} {\partial t_0}=0$ as in the classical case.

We multiply by $m$, $mv$ and $\frac{1}{2}mv^2$ on both sides of Eq.~\eqref{zeroorder}, and integrate $v$ from $-\infty$ to $\infty$, respectively, and obtain
\begin{eqnarray}
    \frac{\partial \rho} {\partial t_0}=0,
 \label{zeroorder1}
\end{eqnarray}
\begin{eqnarray}
    \frac{\partial (\rho u)} {\partial t_0}=0,
 \label{zeroorder2}
\end{eqnarray}
\begin{eqnarray}
    \frac{1}{2} \frac{\partial (\rho K_B T)}{\partial t_0}+\frac{1}{2} \frac{\partial (\rho u^2)}{\partial t_0}=0.
 \label{zeroorder3}
\end{eqnarray}
Since $\rho K_B T+\rho u^2=E$, Eq.~\eqref{zeroorder3} can also be written in the form
\begin{eqnarray}
 \frac{\partial E}{\partial t_0}=0,
 \label{zeroorder3-2}
\end{eqnarray}

Analogously,
we multiply by $m$, $mv$ and $\frac{1}{2}mv^2$ on both sides of Eq.~\eqref{1order}, and integrate $v$ from $-\infty$ to $\infty$, respectively, and obtain
\begin{eqnarray}
    \frac{\partial \rho} {\partial t_1}+\frac{\partial (\rho u)} {\partial x_1}=0,
 \label{1order1}
\end{eqnarray}
\begin{eqnarray}
    \frac{\partial (\rho u)} {\partial t_1}+\frac{\partial E} {\partial x_1}=0,
 \label{1order2}
\end{eqnarray}
\begin{eqnarray}
    \frac{\partial E}{\partial t_1}+\frac{\partial u(E+P)}{\partial x_1}=0,
 \label{1order3}
\end{eqnarray}
In the above three equations we all used that the integrals for the high order term $f^{(1)}$ on the right hand of Eq.~\eqref{1order} are $0$.

We take the Laplace transform of Eq.~\eqref{1order} and find
\begin{eqnarray}
     &&L\bigg[\frac{\partial f^{(0)}} {\partial t_1}+v\frac{\partial f^{(0)}} {\partial x_1}\bigg]=-\frac{1}{\Gamma(1-\beta)\tau_0^{\beta}}(s+\alpha)^{1-\beta} L(f^{(1)}).
 \label{1orderlapalace}
\end{eqnarray}
Thus, one has
\begin{eqnarray}
      L(f^{(1)})= -\Gamma(1-\beta)\tau_0^{\beta}(s+\alpha)^{\beta-1}L\bigg[\frac{\partial f^{(0)}} {\partial t_1}+v\frac{\partial f^{(0)}} {\partial x_1}\bigg].
 \label{1orderlaplace2}
\end{eqnarray}
We multiply by $mv^2$ on both sides of Eq.~\eqref{1orderlaplace2}, and integrate $v$ from $-\infty$ to $\infty$,  and obtain
\begin{eqnarray}
      L(p^{(1)})= \frac{1}{\tau}\Gamma(1-\beta)\tau_0^{\beta}(s+\alpha)^{\beta-1}L(\tilde{P}^{(1)}),
 \label{p1fractional}
\end{eqnarray}
where $p^{(1)}=\int mv^2 f^{(1)}dv$, $\tilde{p}^{(1)}=\int mv^2 \tilde{f}^{(1)}dv$ and  $\tilde{f}^{(1)}=-\tau[\frac{\partial f^{(0)}} {\partial t_1}+v\frac{\partial f^{(0)}} {\partial x_1}]$.

Analogously,
we multiply by $\frac{1}{2}mv^3$ on both sides of Eq.~\eqref{1orderlaplace2}, and integrate $v$ from $-\infty$ to $\infty$,  and obtain
\begin{eqnarray}
      L(q)= \frac{1}{\tau}\Gamma(1-\beta)\tau_0^{\beta}(s+\alpha)^{\beta-1}L(\tilde{q}),
 \label{qfractional}
\end{eqnarray}
where $q=\frac{1}{2}\int mv^3 f^{(1)}dv$ and  $\tilde{q}=\frac{1}{2}\int mv^3 \tilde{f}^{(1)}dv$.

 We multiply by $m$, $mv$ and $\frac{1}{2}mv^2$ on both sides of Eq.~\eqref{2order}, and integrate $v$ from $-\infty$ to $\infty$, respectively, and obtain
\begin{eqnarray}
    \frac{\partial \rho} {\partial t_2}+\frac{\partial (\rho u)} {\partial x_1}=0,
 \label{2order1}
\end{eqnarray}
\begin{eqnarray}
    \frac{\partial (\rho u)} {\partial t_2}+\frac{\partial p^{(1)}} {\partial x_1}=0,
 \label{2order2}
\end{eqnarray}
\begin{eqnarray}
    \frac{\partial E}{\partial t_2}+\frac{\partial q}{\partial x_1}=0.
 \label{2order3}
\end{eqnarray}
Furthermore, multiplying by $\epsilon$ of both sides of Eqs. \eqref{1order1}--\eqref{1order3}, and combining with $\frac{\partial f}{\partial t}=\frac{\partial f^{(0)}}{\partial t_0}+\epsilon \frac{\partial f^{(0)}}{\partial t_1}$ and Eqs.~\eqref{zeroorder1}, \eqref{zeroorder2} and \eqref{zeroorder3-2},
we obtain the classical
Euler equations
\begin{eqnarray}
    \frac{\partial \rho} {\partial t_1}+\frac{\partial (\rho u)} {\partial x}=0,
 \label{euler1}
\end{eqnarray}
\begin{eqnarray}
    \frac{\partial (\rho u)} {\partial t_1}+\frac{\partial E} {\partial x}=0,
 \label{euler2}
\end{eqnarray}
\begin{eqnarray}
    \frac{\partial E}{\partial t_1}+\frac{\partial u(E+P)}{\partial x}=0.
 \label{euler3}
\end{eqnarray}
Analogously, we multiply by $\epsilon$ of both sides of Eqs. \eqref{1order1}--\eqref{1order3}, and by $\epsilon^2$ of both sides of Eqs. \eqref{2order1}--\eqref{2order3},  and combine with Eqs.~\eqref{Chapman-Enskogassumpttion2}, \eqref{Chapman-Enskogassumpttion3}, \eqref{zeroorder1}, \eqref{zeroorder2} and \eqref{zeroorder3-2}, and obtain the fractional Navier--Stokes equations \eqref{NS1}, \eqref{NS2} and \eqref{NS3}.
Here, in the derivation of Eqs.~\eqref{NS2} and \eqref{NS3} Eqs.~\eqref{RLgeneratedlaplace}, \eqref{p1fractional} and \eqref{qfractional} are also used.
\end{proof}

\section{Conclusion}\label{sec-conclusion}

The macroscopic Navier--Stokes equations, formulated in the 19th Century, are still not solved now, which is one of the hard Millennium problems proposed by CMI. In Hilbert's sixth problem, he suggests a method by deriving the Boltzmann equation as an intermediate step from the microscopic Newton’s laws and the probability theory.
In this paper we use the  collision  renewal process to derive the corresponding microscopic master equation \eqref{master} for the time evolution of the probability of the state tensor of positions and velocities of all particles, and then derive the mesoscopic generalized rate equation \eqref{rateequation} for one particle, whose special cases can lead to the Boltzmann equation \eqref{Boltzmannequation}, the classical BGK equation \eqref{rateeqautionsimpleexponent} and the fractional BGK equation \eqref{fractionalBGKequation}. Furthermore, based on the BGK equations we obtain the Euler equations \eqref{euler1}--\eqref{euler3} and the fractional Navier--Stokes equations \eqref{NS1}--\eqref{NS3} by using Chapman--Enskog Method.  Since the solution of the stochastic simulation of trajectories of the collision renewal process can be obtained according to Sec.~\ref{simulationsteps}, we actually find an equivalent stochastic method to solve the classical and fractional NS equations. Finally,  we also establish trajectory counter-part for the phenomenological fractional BGK equation predicted in \cite{G2017}  for the force $F=0$ using the collision renewal process.

There are some further problems worthy of investigation, such as the blow--up properties of the energy and the velocity of fractional NS equations based on the collision renewal process, the fractional NS equations in the space-- and time--dependent force field, and so on.

\backmatter
\section*{Abbreviations}
Navier--Stokes: NS;\\
Clay Mathematics Institute: CMI;\\
continuous time random walk: CTRW; \\
independent and identically distributed: i.i.d.;\\
probability density function: PDF.\\

\section*{Declarations}

The authors declare that there are no conflict of interests.

\section*{Data availability}
No data was used for the research described in the paper.


\bibliography{sn-bibliography}

\end{document}